\documentclass[a4paper,12pt]{amsart}

\usepackage{geometry}
\usepackage[comma,numbers]{natbib}
\usepackage[inline]{enumitem}
\usepackage{graphicx}
\usepackage[textfont=it,tableposition=above,font=small]{caption}		

\usepackage[foot]{amsaddr}	
\usepackage{url}						

\usepackage{amssymb}				

\usepackage{silence}
\usepackage{xcolor}
\RequirePackage[colorlinks,citecolor=blue,urlcolor=blue]{hyperref}
\usepackage[norefs,nocites]{refcheck}

\newtheorem{theorem}{Theorem}
\newtheorem{lemma}[theorem]{Lemma}

\newtheorem{cor}[theorem]{Corollary}

\theoremstyle{definition}
\newtheorem*{definition}{Definition}
\newtheorem{example}{Example}

\newcommand{\twofig}{.45}                             
\newcommand{\twofigb}{.4}                             
\newcommand{\Meas}[1][\R^d]{\mathcal M\left(#1\right)}
\newcommand{\Atom}[1][\R^d]{\mathcal A\left(#1\right)}
\newcommand{\D}{D}																		
\newcommand{\R}{\mathbb R}														
\newcommand{\tr}{^\mathsf{T}}													
\newcommand{\Sph}[1][d]{\mathbb{S}^{#1-1}}						
\newcommand{\SphO}{\mathbb{S}^{1}}										

\DeclareMathOperator{\intrOp}{int}									  
\DeclareMathOperator{\clOp}{cl}   									  
\DeclareMathOperator{\bdOp}{bd}												
\DeclareMathOperator{\relintOp}{relint}								
\DeclareMathOperator{\relbdOp}{relbd} 								
\DeclareMathOperator{\relclOp}{relcl}									

\DeclareMathOperator{\affOp}{aff}
\newcommand{\intr}[1]{\intrOp\left(#1\right)}
\newcommand{\cl}[1]{\clOp\left(#1\right)}
\newcommand{\bd}[1]{\bdOp\left(#1\right)}
\newcommand{\relint}[1]{\relintOp\left(#1\right)}
\newcommand{\relcl}[1]{\relclOp\left(#1\right)}
\newcommand{\relbd}[1]{\relbdOp\left(#1\right)}

\newcommand{\half}{\mathcal H}												
\newcommand{\halfo}{\half^\circ}                      
\newcommand{\flag}{\mathcal F}												

\newcommand{\aff}[1]{\affOp\left(#1\right)}                    

\newcommand{\medianOp}{D^*}														
\newcommand{\median}[1][\mu]{\medianOp\left(#1\right)}

\newcommand{\Damu}[1][\mu]{D_\alpha(#1)}
\newcommand{\am}[1][\mu]{\alpha^*(#1)}

\newcommand{\depth}[2]{\D\left(#1;#2\right)}

\newcommand{\pkg}[1]{{\normalfont\fontseries{b}\selectfont #1}}
\let\proglang=\textsf

\title[Another look at halfspace depth: Flag halfspaces]{Another look at halfspace depth: \\ Flag halfspaces with applications}

\author{Du\v{s}an Pokorn\'{y}}
\author{Petra Laketa}
\author{Stanislav Nagy}

\email{nagy@karlin.mff.cuni.cz}

\address{
	Charles University,
	Faculty of Mathematics and Physics,
	Prague, Czech Republic
}

	
\date{\today}

\begin{document}

\begin{abstract}
The halfspace depth is a well studied tool of nonparametric statistics in multivariate spaces, naturally inducing a multivariate generalisation of quantiles. The halfspace depth of a point with respect to a measure is defined as the infimum mass of closed halfspaces that contain the given point. In general, a closed halfspace that attains that infimum does not have to exist. We introduce a flag halfspace --- an intermediary between a closed halfspace and its interior. We demonstrate that the halfspace depth can be equivalently formulated also in terms of flag halfspaces, and that there always exists a flag halfspace whose boundary passes through any given point $x$, and has mass exactly equal to the halfspace depth of $x$. Flag halfspaces allow us to derive theoretical results regarding the halfspace depth without the need to differentiate absolutely continuous measures from measures containing atoms, as was frequently done previously. The notion of flag halfspaces is used to state results on the dimensionality of the halfspace median set for random samples. We prove that under mild conditions, the dimension of the sample halfspace median set of $d$-variate data cannot be $d-1$, and that for $d=2$ the sample halfspace median set must be either a two-dimensional convex polygon, or a data point. The latter result guarantees that the computational algorithm for the sample halfspace median form the \proglang{R} package \pkg{TukeyRegion} is exact also in the case when the median set is less-than-full-dimensional in dimension $d=2$.
\end{abstract}

\maketitle

\section{Introduction: Halfspace depth and its median}	

Denote by $\Meas$ the set of all finite Borel measures on the Euclidean space $\R^d$. The \emph{halfspace} (or \emph{Tukey}) \emph{depth} of $x \in \R^d$ with respect to (w.r.t.) $\mu \in \Meas$ is defined as\footnote{We consider the halfspace depth w.r.t. finite measures $\mu$, that is when $\mu(\R^d) < \infty$. Compared to the usual setup of probability measures, this extension is minor, and made only for notational convenience. All our results could be considered also for probability measures only, with obvious modifications.}
	\begin{equation}	\label{halfspace depth}
	\depth{x}{\mu} = \inf \left\{ \mu(H) \colon H \in \half(x) \right\},	
	\end{equation}
where $\half(x)$ is the collection of closed halfspaces in $\R^d$ that contain $x$ on their boundary. The halfspace depth quantifies the centrality of $x$ w.r.t. the mass of $\mu$. That is quite useful in nonparametric statistics, as it allows us to rank sample points according to their depth, from the central to the peripheral ones. As such, the depth enables the introduction of rankings, orderings, and quantile-like inference to multivariate datasets \cite{Donoho_Gasko1992, Tukey1975, Zuo_Serfling2000}. The upper level sets of the halfspace depth of $\mu$, given for $\alpha \geq 0$ by
	\begin{equation}	\label{central region}
	\Damu = \left\{ x \in \R^d \colon \depth{x}{\mu} \geq \alpha \right\},
	\end{equation}
play in nonparametric statistics the role of the inner quantile regions of $\mu$. They are often called the \emph{(halfspace) central regions} of $\mu$. The sets \eqref{central region} are nested, closed and convex; they are compact for $\alpha > 0$, and non-empty for $\alpha \leq \am$, where $\am = \sup_{x \in \R^d} \depth{x}{\mu}$ is the maximum halfspace depth of $\mu$. Of special importance is the set $\median = D_{\alpha^*(\mu)}$, which contains points that are the most centrally positioned w.r.t. $\mu$. It is called the set of the \emph{halfspace medians} of $\mu$ and, as its name suggests, it generalises the median to $\R^d$. The halfspace depth has many applications in multivariate statistics, and is already for 30 years a subject of active research \cite{Liu_etal2019, Masse2004, Mizera_Volauf2002, Nagy_etal2019, Rousseeuw_Ruts1999}. Although many other statistical depth functions have been developed \cite{Chernozhukov_etal2017, Mosler_Mozharovskyi2022, Zuo_Serfling2000}, in this paper we focus on the halfspace depth, and sometimes write simply \emph{depth} instead of \emph{halfspace depth}. 
We call a measure $\mu \in \Meas$ \emph{smooth} if the $\mu$-mass of every hyperplane in $\R^d$ is zero. A measure with a density is smooth; examples of non-smooth measures are those with an atom. The infimum in \eqref{halfspace depth} is attained for smooth measures. That is why theoretical results on the halfspace depth are often formulated only for smooth measures, and why the analysis of the sample halfspace depth (that is, the halfspace depth evaluated w.r.t. empirical measures of random samples) is performed using different techniques \cite{Liu_etal2020, Masse2004, Mizera_Volauf2002}. In this paper we introduce \emph{flag halfspaces} --- symmetrised variants of closed halfspaces that may be considered in \eqref{halfspace depth} instead of $\half(x)$ without altering the depth, with the property that a flag halfspace attaining the depth always exists. We will see that our restatement of formula~\eqref{halfspace depth} simplifies many theoretical derivations about the halfspace depth, as it is no longer needed to distinguish whether the infimum in~\eqref{halfspace depth} is attained.

Flag halfspaces are introduced in Section~\ref{section:flag halfspace}. 
Two applications to the computation of the depth are given in Section~\ref{section:depth for empirical measures}. In Section~\ref{section:dimensionality}, we investigate the dimensionality and the structure of the median set $\median$ for $\mu$ an empirical measure. We show that for datasets sampled from absolutely continuous probability measures in $\R^d$, the halfspace median set cannot be of dimension $d-1$, almost surely. In a series of examples in $\R^3$ we demonstrate that already for random samples of size $n = 8$ from the standard Gaussian distribution, halfspace median sets of dimensions $0$, $1$, and $3$ occur with positive probability. In Section~\ref{section:d2} we deal with the special situation of data of dimension $d=2$. We show that if the dataset satisfies a mild condition of general position, then the halfspace median set must be either a full-dimensional polygon, or a data point. Both these advances find applications in the computation of the halfspace median and the central regions \eqref{central region}, where the dimensionality of $\median$ plays a crucial role \cite{Fojtik_etal2022, Liu_etal2019}. The paper is complemented by online Supplementary Material containing \proglang{R} and \proglang{Mathematica} scripts with visualisations and computations completing examples from Section~\ref{section:depth for empirical measures}.

\subsection*{Notations.} Some of our proofs are based on convexity theory. As a basic reference we take \cite{Schneider2014}; we now gather notations and elementary definitions that will be used throughout the paper. The unit sphere in $\R^d$ is $\Sph$. We write $S \subset K$ for $S$ being a proper subset of $K$. The restriction of $\mu\in\Meas$ to a Borel set $S\subseteq \R^d$ is denoted by $\mu|_{S} \in \Meas$ and is defined by $\mu|_{S}\left(B\right)=\mu\left(B\cap S\right)$ for $B \subseteq \R^d$ Borel. The affine hull $\aff{S}$ of $S\subseteq \R^d$ is the smallest affine subspace of $\R^d$ containing $S$. The dimension $\dim(S)$ of $S$ is defined as the dimension of $\aff{S}$. For example, the affine hull of two different points in $\R^d$ is the infinite line joining them, and its dimension is 1. We write $\intr{S}$, $\cl{S}$, and $\bd{S}$ for the interior, closure, and boundary of $S \subseteq \R^d$. 
The interior, closure, and boundary of $S$ when considered as a subset of its affine hull $\aff{S}$ is denoted by $\relint{S}$, $\relcl{S}$ and $\relbd{S}$, and is called the relative interior, relative closure, and relative boundary of $S$, respectively. Of course, if $\dim(S) = d$, the interior is the same as the relative interior etc.
	
The class of all closed halfspaces in $\R^d$ is $\half$. A generic halfspace from $\half$ may be denoted simply by $H$; $H_{x,v}$ means a halfspace $\left\{ y \in \R^d \colon \left\langle y, v \right\rangle \geq \left\langle x,v \right\rangle \right\}$ whose boundary passes through $x \in \R^d$ with inner normal $v \in \R^d \setminus \{0\}$. For an affine space $A\subseteq \R^d$ and $x \in A$ we denote by $\half(x,A)$ the set of all relatively closed halfspaces $H$ in $A$ whose relative boundary contains $x$; surely $\half(x,\R^d) \equiv \half(x)$. We say that a sequence of halfspaces $\{H_{x_n,v_n}\}_{n=1}^\infty\subset\half$ converges to $H_{x,v}\in\half$ if $x_n\rightarrow x$ and $v_n\rightarrow v$. Finally, for any of the symbols $\half$, $\half(x)$, or $\half(x,A)$, a superscript $\circ$ designates the corresponding relatively open halfspaces, e.g. $\halfo(x,A) = \left\{ \relint{H} \colon H \in \half(x,A) \right\}$.

\section{Flag halfspaces} \label{section:flag halfspace}

For $\mu\in\Meas$ and $x \in \R^d$ we call $H \in \half(x)$ a \emph{minimising halfspace} of $\mu$ at $x$ if $\mu(H) = \depth{x}{\mu}$. For $d = 1$ minimising halfspaces always trivially exist. They also exist if $\mu$ is smooth, or if $\mu$ is supported in a finite number of points. In general, however, the infimum in~\eqref{halfspace depth} does not have to be attained. We give a simple example.

\begin{example}	\label{example:flag}
Take $\mu \in \Meas[\R^2]$ the sum of the Dirac measure at $a = (1,1) \in \R^2$ and the uniform distribution on the disk $\left\{ x \in \R^2 \colon \left\Vert x \right\Vert \leq 2 \right\}$. For $x = (1,0) \in \R^2$ no minimising halfspace exists. As we see in Figure~\ref{figure:flag halfspace}, the depth $\D(x;\mu)$ is approached by $\mu(H_{x,v_n})$ for a sequence of halfspaces $H_n \equiv H_{x,v_n}$, $n=1,2,\dots$, with inner normals $v_n = \left( \cos(-1/n), \sin(-1/n) \right)$ that converge to $v = (1,0) \in \SphO$, yet $\D(x;\mu) = \lim_{n \to \infty} \mu(H_{x,v_n}) < \mu(H_{x,v})$.
\end{example}

\begin{figure}[htpb]
\includegraphics[width=\twofigb\textwidth]{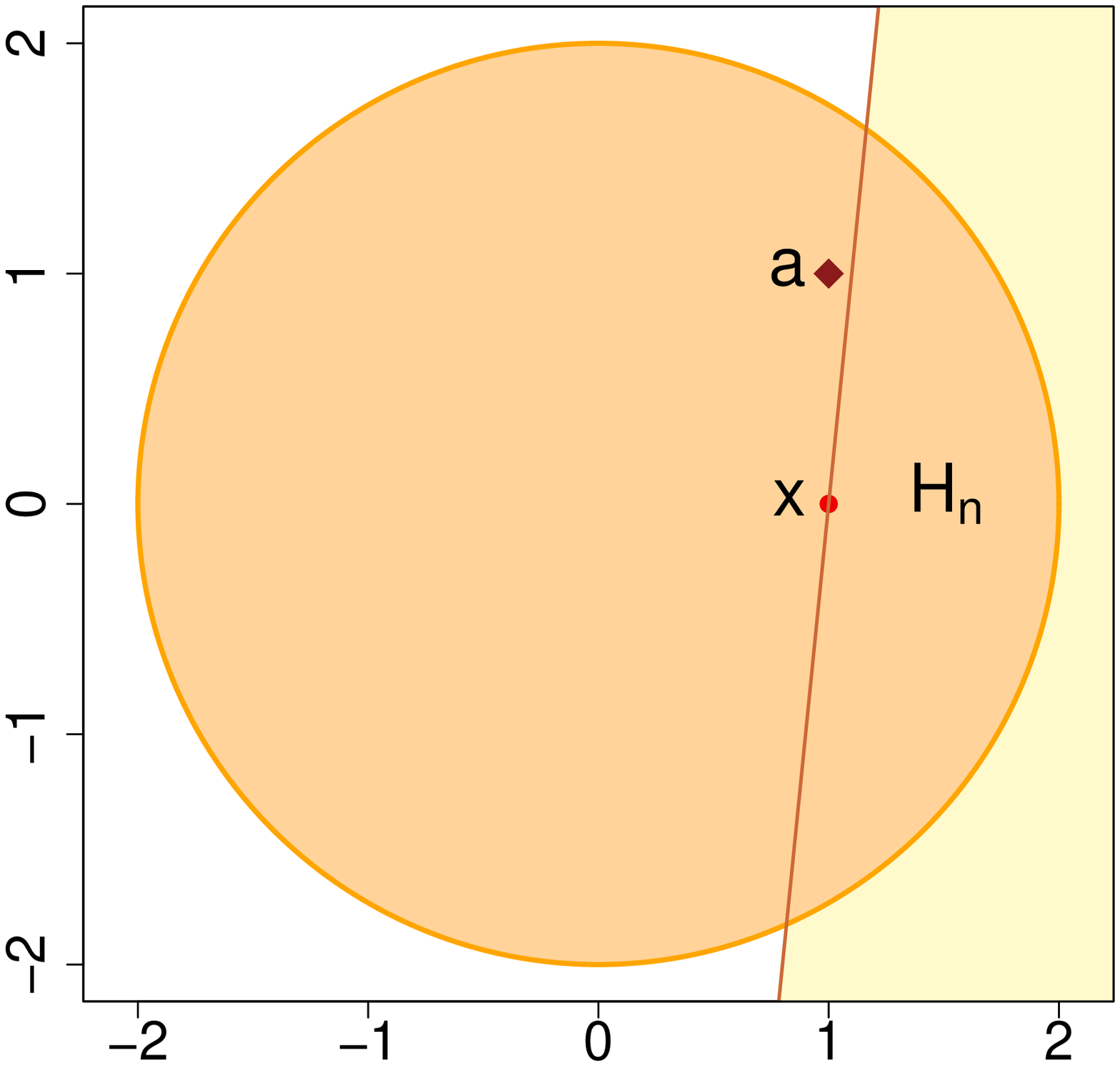} \qquad \includegraphics[width=\twofigb\textwidth]{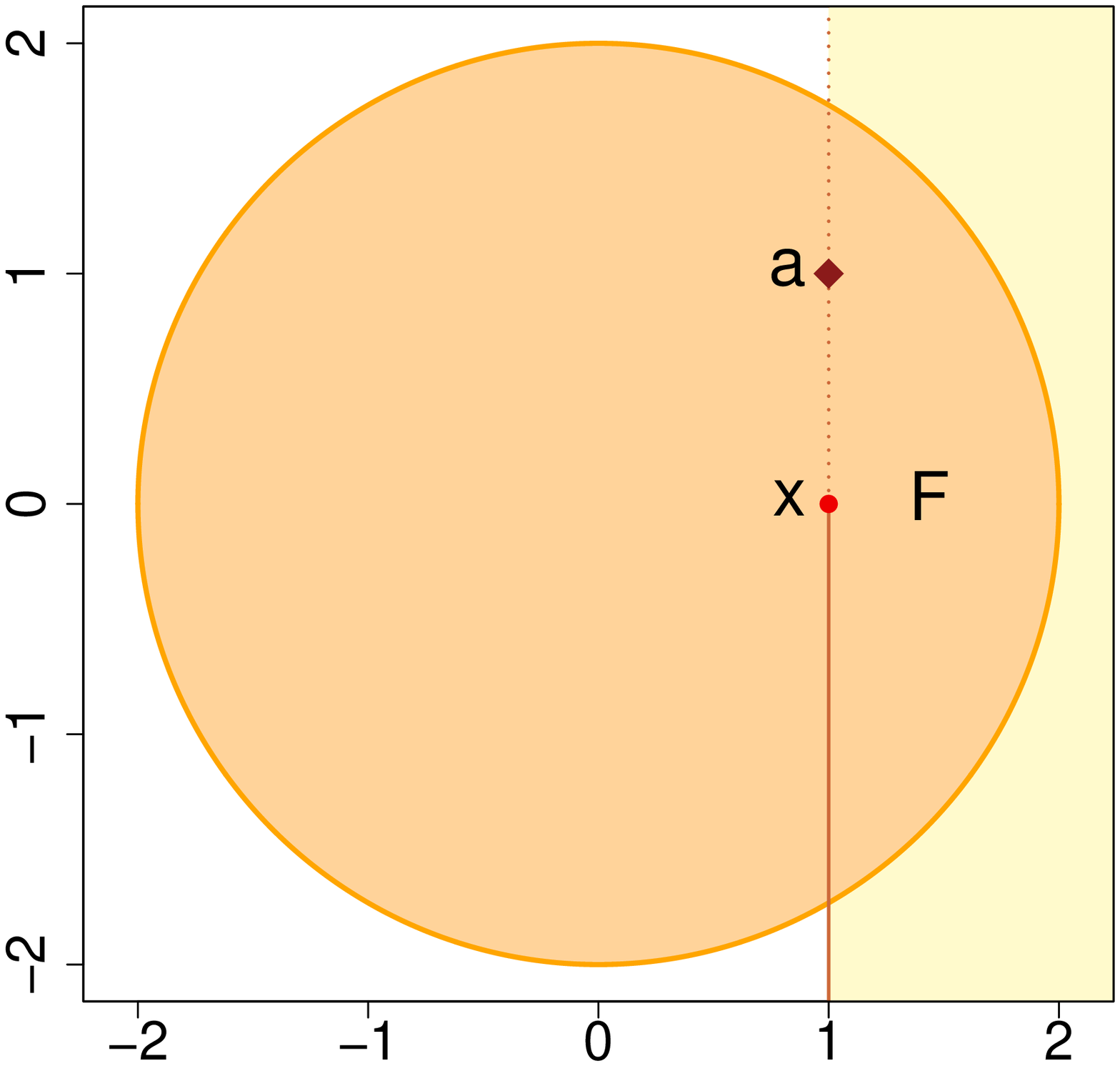}%
\caption{The support of $\mu \in \Meas[\R^2]$ from Example~\ref{example:flag} (coloured disk) and its atom $a$ (diamond). No minimising halfspace of $\mu$ at $x = (1,0)\in\R^2$ (coloured point) exists. In the left hand panel we see a halfspace $H_n \in \half(x)$ whose $\mu$-mass is almost $\D(x;\mu)$. It does not contain $a$. In the right hand panel the minimising flag halfspace $F \in \flag(x)$ of $\mu$ at $x$ is displayed.}
\label{figure:flag halfspace}
\end{figure}

The problem with measures not attaining the infimum in \eqref{halfspace depth} is elegantly resolved by considering flag halfspaces instead of the usual closed halfspaces. 

\begin{definition}
Define $\flag(x)$ as the system of all sets $F$ of the form 
	\begin{equation}	\label{eq:flag halfspace}
	F = \{x\}\cup \left(\bigcup_{k=1}^d G_k\right)	
	\end{equation}
where $G_d\in\halfo(x)$, and $G_k \in \halfo(x,\relbd{G_{k+1}})$ for every $k = 1, \dots, d-1$. Any element of $\flag(x)$ is called a \emph{flag halfspace} at $x$. 
\end{definition}

The name \emph{flag} comes from geometry \cite{Schneider2014}, where an analogous recursive construction is considered, involving nested faces of convex polytopes. 
The formal definition of flag halfspaces is somewhat convoluted, but these sets appear naturally. In $\R^2$, a flag halfspace at $x$ is the union of an open halfplane $G_2$ whose boundary passes through $x$, a relatively open halfline $G_1$ originating at $x$ contained in the one-dimensional affine space (line) $\bd{G_2}$, and the $0$-dimensional point $x$ itself. For an example see Figure~\ref{figure:flag halfspace}. A flag halfspace is neither an open nor a closed set. In contrast to a usual closed halfspace, a complement of a flag halfspace $F \in \flag(x)$ is, except for its central point $x$, again a flag halfspace from $\flag(x)$, i.e. $(\R^d \setminus F) \cup \{ x \} \in \flag(x)$. Several more interesting properties and characterisations of flag halfspaces can be found in~\cite{Laketa_etal2022simple}.

We define a \emph{minimising flag halfspace} of $\mu$ at $x$ to be any $F \in \flag(x)$ that satisfies $\mu(F) = \D(x;\mu)$. In the following Theorem~\ref{theorem:Pokorny} we show that the halfspace depth \eqref{halfspace depth} of any measure can be expressed in terms of the $\mu$-mass of flag halfspaces, and a minimising flag halfspace always exists. The intuition behind this result is as follows: Even if the minimising closed halfspace of $x\in\R^d$ does not exist, there is a sequence of closed halfspaces $\{H_n\}_{n=1}^\infty \subset \half(x)$ that satisfies 
	\begin{equation}\label{minimising limit}
	\lim_{n\to\infty}\mu\left(H_n\right)=\depth{x}{\mu}.
	\end{equation}
Because the unit normals $\{v_n\}_{n=1}^\infty$ of these halfspaces come from the compact set $\Sph$, we can also assume that the sequence of halfspaces is convergent and $\lim_{n \to \infty} v_n = v \in \Sph$ (otherwise, we extract a convergent subsequence). For $n$ large enough, $\mu\left(H_n\right)$ is arbitrarily close to $\depth{x}{\mu}$, but this fact alone, of course, does not imply that the $\mu$-mass of the limit $H \equiv H_{x,v}$ defined as $H = \lim_{n\to\infty}H_n$ is equal to $\depth{x}{\mu}$. It turns out that for general measures, it is not possible to find any useful upper bound on the mass $\mu\left(H\right)$, but it is possible to bound the mass of its interior by $\mu\left(\intr{H}\right)\leq\depth{x}{\mu}$. The interior of $H$ is the first open halfspace $G_d$ in the construction of the minimising flag halfspace \eqref{eq:flag halfspace}. The remaining relatively open halfspaces $G_k$ are found by iterating the same process inside the relative boundary of the previous $G_{k+1}$, $k=1,\dots,d-1$. In the right hand panel of Figure~\ref{figure:flag in plane} we see a visualisation of our setup, with $d=2$. In the situation displayed, as $n \to \infty$, the halfspaces $H_n$ do not intersect the halfline $G_1^- \subset \bd{H}$ originating at $x$, so the $\mu$-mass of $G_1^-$ does not contribute to the depth of $x$, and $G_1^-$ should not be contained in $F$. The formal statement of our theorem follows. 

\begin{figure}[htpb]
\includegraphics[width=\twofigb\textwidth]{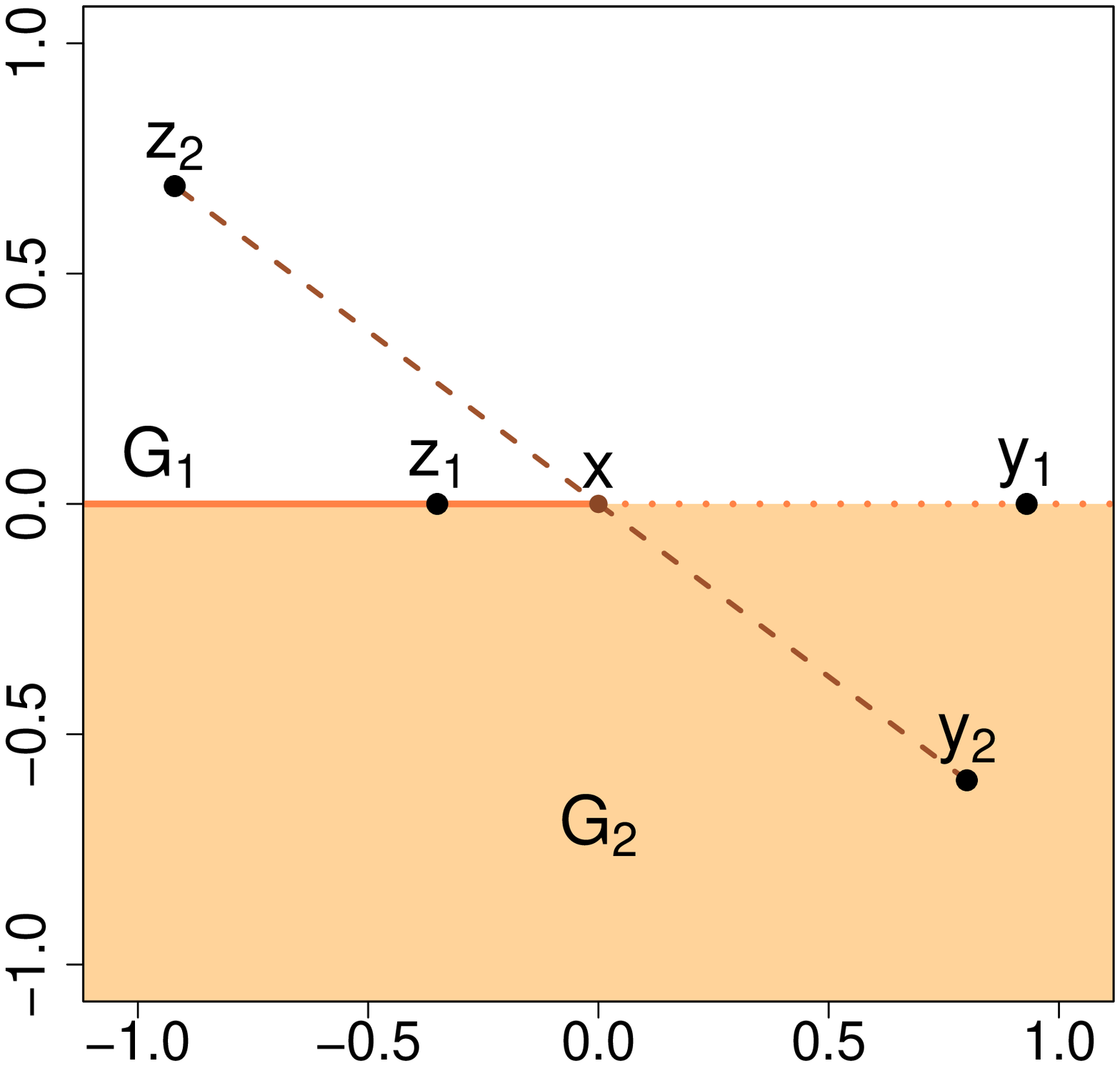} \qquad \includegraphics[width=\twofigb\textwidth]{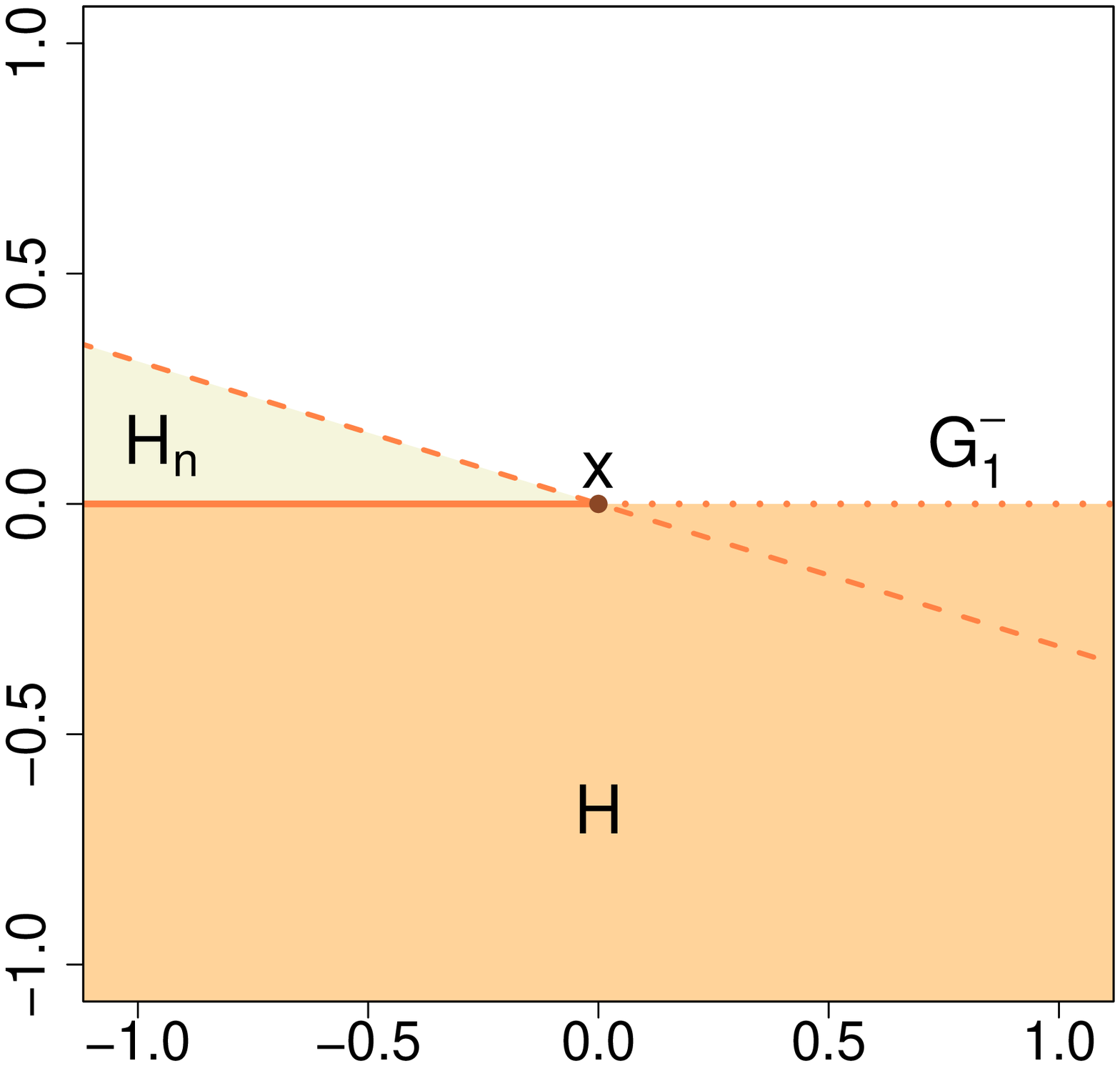}
\caption{Left hand panel: A flag halfspace $F=\{x\}\cup G_1\cup G_2 \in \flag(x)$ in the plane. For any line segment with endpoints $y, z$ passing through $x$, exactly one of the points $y,z$ belongs to $F$. Right hand panel: $H=\cl{F}$ for $H$ being the limit of a sequence of closed halfspaces $\left\{ H_n \right\}_{n=1}^\infty$ satisfying \eqref{minimising limit}. Certainly, $\mu\left(H\right)$ is not necessarily equal to $\depth{x}{\mu}$ because it is possible that $\mu\left(H\setminus H_n\right)\geq\mu\left(G_1^-\right)>0$.} 
\label{figure:flag in plane}
\end{figure}

\begin{theorem}	\label{theorem:Pokorny}
For any $\mu \in \Meas$ and $x \in \R^d$ we have
	\begin{equation}	\label{flag halfspace depth}
	\D\left(x;\mu\right) = \min\left\{ \mu\left(F\right) \colon F \in \flag(x) \right\}.
	\end{equation}
In particular, there always exists a minimising flag halfspace. 
\end{theorem}
 
\begin{proof}
Let $\left\{H_n\right\}_{n=1}^\infty \subset \half(x)$ be a sequence of halfspaces satisfying \eqref{minimising limit} with limit $H \equiv H_{x,v} =\lim_{n\to\infty}H_n$. For all $n=1,2,\dots$ we have
	\begin{equation}	\label{mu decomposition}
	\mu(H_n) \geq \mu(H_n \cap H) = \mu\left(H_n \cap \intr{H}\right) + \mu\left(H_n \cap \bd{H} \right).	
	\end{equation}
We first bound both summands on the right hand side from below. For each $n=1,2,\dots$ we define $A_n= \left(\bigcap_{m\geq n}H_m\right) \cap \intr{H} \subseteq H_n \cap \intr{H}$. From the convergence of the halfspaces $\left\{ H_n \right\}_{n=1}^\infty$ we know that $A_n\uparrow\intr{H}$ as $n\rightarrow \infty$, and using the continuity of measure from below \cite[Theorem~3.1.11]{Dudley2002} we obtain the equality in
	\begin{equation}\label{intersection mass}
	\mu(\intr{H}) = \lim_{n\to\infty} \mu\left( A_n \right) \leq \liminf_{n\to\infty} \mu\left( H_n \cap \intr{H} \right).
	\end{equation}
On the other side, $x\in H_n\cap \bd{H}$ for all $n=1,2,\dots$, so $H_n \cap \bd{H}$ is either a closed halfspace when considered in the $(d-1)$-dimensional space $\bd{H}$, or is equal to $\bd{H}$. In any case, we have that $ \mu\left(H_n \cap \bd{H} \right)\geq \depth{x}{\mu|_{\bd{H}}}$ for $\mu|_{\bd{H}}$ the restriction of $\mu$ to the hyperplane $\bd{H}$. Consequently
	\begin{equation}\label{boundary mass}
	\liminf_{n\to\infty} \mu\left(H_n \cap \bd{H} \right)\geq \depth{x}{\mu|_{\bd{H}}}.
	\end{equation}
Combining~\eqref{mu decomposition}, \eqref{intersection mass} and~\eqref{boundary mass} one gets
	\begin{equation}	\label{boundary depth}	
	\begin{aligned}
	\depth{x}{\mu} & = \lim_{n \to \infty} \mu(H_n) \geq \liminf_{n\to\infty} \mu\left( H_n \cap \intr{H} \right)+ \liminf_{n\to\infty} \mu\left(H_n \cap \bd{H} \right) \\
	& \geq \mu(\intr{H}) + \depth{x}{\mu|_{\bd{H}} }.
	\end{aligned}
	\end{equation}
Assume now for a contradiction that the inequality in~\eqref{boundary depth} is strict, i.e. that $\depth{x}{\mu}- \mu(\intr{H}) - \depth{x}{\mu|_{\bd{H}}}=c>0$. The definition of the halfspace depth implies that there exists a halfspace $\widetilde{H}\in\half(x,\bd{H})$ in the hyperplane $\bd{H}$ that satisfies
	\begin{equation}\label{minimising boundary}
	\mu|_{\bd{H}}(\widetilde{H})<\depth{x}{\mu|_{\bd{H}} }+c/2.
	\end{equation}
Denote by $\widetilde{v} \in \Sph$ the unit inner normal of $\widetilde{H}$ and set $w_n=v+\widetilde{v}/n$ and $C_n=H_{x,w_n}\setminus H$ for $n=1,2,\dots$. Then $w_n\to v$ and $C_n\downarrow \emptyset$ as $n\to\infty$, meaning that $\lim_{n\to\infty} H_{x,w_n} = H$ and $\lim_{n\to\infty}\mu(C_n)=0$ due to the continuity of measure from above \cite[Theorem~3.1.1]{Dudley2002}. For $n$ large enough we have $\mu(H_{x,w_n}\setminus H)<c/2$. Note also that $H_{x,w_n}\cap \bd{H}=\widetilde{H}$ for all $n=1,2,\dots$, due to the choice of $w_n$. Altogether, we have
	\begin{equation}	\label{contradiction}
	\begin{aligned}
	\mu(H_{x,w_n}) & = \mu(H_{x,w_n}\cap \intr{H})+\mu(H_{x,w_n}\cap \bd{H})+ \mu(H_{x,w_n}\setminus H) \\
	& <\mu(\intr{H})+\mu(\widetilde{H})+c/2 = \mu(\intr{H})+\mu|_{\bd{H}}(\widetilde{H})+c/2 \\
	&<\mu(\intr{H})+\depth{x}{\mu|_{\bd{H}} }+c=\depth{x}{\mu},
	\end{aligned}
	\end{equation}
where the last inequality in~\eqref{contradiction} follows from~\eqref{minimising boundary}. Note that because $H_{x,w_n} \in \half(x)$, inequality \eqref{contradiction} contradicts the definition of the halfspace depth \eqref{halfspace depth}, and we get
	\begin{equation} \label{induction step}
	\depth{x}{\mu}= \mu(G_d) + \depth{x}{\mu|_{\relbd{G_d}}},
	\end{equation}
where we denoted $G_d=\intr{H}\in\halfo(x)$. We have just constructed the first open halfspace $G_d$ in the system \eqref{eq:flag halfspace}. We proceed by induction. We consider $\mu|_{\bd{H}} = \mu|_{\relbd{G_d}}$ instead of $\mu$ and using the same argument obtain $G_{d-1}\in\halfo(x,\relbd{G_d})$ that satisfies an equation analogous to \eqref{induction step}, i.e. $\depth{x}{\mu|_{\relbd{G_d}}}=\mu(G_{d-1})+\depth{x}{\mu|_{\relbd{G_{d-1}}}}$. Continuing the same procedure we eventually obtain a flag halfspace $F=\{x\}\cup\left(\bigcup_{k=1}^d G_k\right)$ such that 
    \begin{equation*}	
	\begin{aligned}
	\depth{x}{\mu} &= \mu(\intr{H}) + \depth{x}{\mu|_{\bd{H}}}=\mu(G_d)+\mu(G_{d-1})+\depth{x}{\mu|_{\relbd{G_{d-1}}}}=\dots\\
	& =\sum_{k=2}^{d}\mu(G_k)+\depth{x}{\mu|_{\relbd{G_{2}}}} =\sum_{k=1}^{d}\mu(G_k)+\mu(\{x\})=\mu(F).
	\end{aligned}
	\end{equation*}
The last but one equality above follows from the fact that $\relbd{G_2}$ is a line, meaning that $G_1$ is one of the two relatively open halflines determined by $x$ in $\relbd{G_2}$ having a smaller $\mu$-mass. Thus, $\depth{x}{\mu|_{\relbd{G_{2}}}}=\mu\left(\{x\}\right)+\mu(G_1)$.
\end{proof}

In Example~\ref{example:flag}, the single minimising flag halfspace of $\mu$ at $x$ is 
	\[	F = \{ x \} \cup \left\{ \left(1, x_2\right) \in \R^2 \colon x_2 < 0 \right\} \cup \left\{ \left(x_1, x_2 \right) \in \R^2 \colon x_1 > 1 \right\} \in \flag(x).	\]

In formula \eqref{induction step} in the proof of Theorem~\ref{theorem:Pokorny} we unveiled the recursive nature of the halfspace depth. The following result formalises that observation. In the special situation of an empirical measure $\mu \in \Meas$, a related result has been observed in \cite[Theorems~1 and~2]{Dyckerhoff_Mozharovskyi2016} and successfully applied in the task of exact computation of the halfspace depth.

\begin{cor}\label{cor:boundary depth}
For $x\in\R^d$ and $\mu\in\Meas$ it holds true that \[\depth{x}{\mu}=\inf_{H\in\half(x)}\left(\mu(\intr{H})+\depth{x}{\mu|_{\bd{H}}}\right).\]
\end{cor}

\begin{proof}
There are more flag halfspaces in $\flag(x)$ than closed halfspaces in $\half(x)$, in the sense that the mapping $\flag(x) \to \half(x) \colon F \mapsto \cl{F}$ is not bijective. We define an equivalence relation $\sim$ between the elements of $\flag(x)$ by
	\begin{equation*}
	 F_1\sim F_2\mbox{ if and only if }\cl{F_1}=\cl{F_2}.
	\end{equation*}
By $\mathcal{K}$ we denote the quotient set of $\sim$. This allows us to rewrite~\eqref{flag halfspace depth} from Theorem~\ref{theorem:Pokorny} as
	\[	\depth{x}{\mu}=\min_{F\in\flag(x)}\mu(F)=\inf_{K\in\mathcal{K}}\inf_{F\in K}\mu(F).	\]
Note that for flag halfspaces, $\intr{F_1}=\intr{F_2}$ is equivalent with $\cl{F_1}=\cl{F_2}$. Take $K\in\mathcal{K}$ and denote $G_K=\intr{F}\in\halfo(x)$ for $F\in K$. Then each $F\in K$ can be represented as $F=G_K\cup F'$, for $F'$ a flag halfspace centred at $x$ when considered inside the affine space $\bd{G_K}$ (denoted by $F'\in \flag\left(x,\bd{G_K}\right)$). We get, using Theorem~\ref{theorem:Pokorny} again, 
	\[	\inf_{F\in K}\mu(F)=\mu(G_K)+\inf_{F'\in \flag\left(x,\bd{G_K}\right)}\mu(F')=\mu(G_K)+\depth{x}{\mu|_{\bd{G_K}}}.	\]
The mapping $\halfo(x)\to\half(x) \colon G\mapsto \cl{G}$ is a bijection, so any $K\in\mathcal{K}$ corresponds to exactly one element $H=\cl{G_K}\in\half(x)$, and we obtain desired result.
\end{proof}

\section{Applications: Properties of the sample halfspace median}	\label{section:depth for empirical measures}

We now use flag halfspaces to derive several properties of the sample halfspace median that are of interest in the practice of the depth; additional applications of flag halfspaces to the theory of the halfspace depth can be found in \cite{Laketa_Nagy2022, Laketa_etal2022simple}. Write $\Atom$ for the set of all empirical measures $\mu\in\Meas$, that is all purely atomic probability measures with a finite number $n$ of atoms, each atom having $\mu$-mass $1/n$, for some $n=1,2,\dots$. These measures are typically obtained observing a random sample $X_1, \dots, X_n \in \R^d$ from a probability distribution $\nu \in \Meas$, each sample point corresponding to an atom. To approximate the halfspace depth of $\nu$, the depth of $\mu$ is computed. The latter depth function is standardly used for inference about the unknown distribution $\nu$. Naturally, it is therefore crucial to understand the behaviour of the halfspace depth w.r.t. empirical measures. We provide results on the dimensionality of the median set, assuming that the atoms of $\mu \in \Atom$ lie in a sufficiently general position. The last assumption is not restrictive; it is satisfied if, for instance, the measure $\nu$ from which we sample is smooth. The proof of the following lemma is standard and omitted. 

\begin{lemma}\label{random samples}
Let $X_1, X_2,\dots, X_n$ be independent random variables sampled from smooth (and possibly different) probability measures from $\Meas$. Then the following holds true almost surely.
\begin{enumerate}[label=(\roman*), ref=(\roman*)]
	\item \label{general position} The points $X_1, X_2, \dots, X_n$ are in general position.\footnote{A set $S$ of points in $\R^d$ is said to lie in \emph{general position} if no subset of $k$ of these points lies in a $(k-2)$-dimensional affine space, for all $k=2,\dots,d+1$. If there are $n>d$ points in $S$, this is equivalent to saying that no hyperplane in $\R^d$ contains more than $d$ points from $S$.} 
	\item \label{lines intersection} Writing $l(x,y)$ for the infinite line determined by $x\neq y\in \R^d$, if $d\geq 2$ and $k_1,\dots ,k_6\in \{1,2,\dots ,n\}$ are pairwise different indices, then 
					\[	l(X_{k_1},X_{k_2}) \cap l(X_{k_3},X_{k_4}) \cap l(X_{k_5},X_{k_6}) = \emptyset.	\]
	\end{enumerate}
\end{lemma}

\subsection{Dimensionality of the sample halfspace median}  \label{section:dimensionality}

As our first application we show that for an empirical measure with atoms in general position, the median set $\median$ in dimension $d\geq 2$ cannot be $(d-1)$-dimensional, unless we are in the trivial case when the number of atoms is equal to $d$. Our findings should be seen as complementary to the earlier advances from \cite[Lemma~6]{Struyf_Rousseeuw1999}, where it was demonstrated that for $\mu \in \Atom$ with atoms in general position are all the depth regions $\Damu$ full-dimensional, except for possibly the depth median $\median$.  

\begin{theorem}\label{theorem:median dimension}
Let $\mu\in\Atom$ be a measure with $n$ atoms of mass $1/n$ in general position. If $n \neq d \geq 2$, then $\dim(\median) \neq d-1$.
\end{theorem}

\begin{proof}
We use two auxiliary lemmas. Our first lemma is a special case of a more general result that can be found in \cite[Lemma~4]{Laketa_Nagy2021b}. In \cite{Laketa_Nagy2021b}, that lemma is formulated with a final inequality $\mu(\intr{H})\leq\alpha$; for $\mu \in \Atom$ also a strict inequality can be written, because the depth of $\mu$ attains only finitely many values.

\begin{lemma}	\label{lemma:Laketa}
Suppose that $\mu\in\Meas$, $\alpha > 0$, a point $x\notin \Damu$ and a face $F$ of $\Damu$ are given so that the relatively open line segment $L(x,y)$ formed by $x$ and $y$ does not intersect $\Damu$ for any $y \in F$. Then there exists a touching\footnote{A halfspace $H \in \half$ is a \emph{touching} halfspace of a non-empty convex set $A \subset \R^d$ if $H \cap \cl{A} \ne \emptyset$ and $\intr{H} \cap A = \emptyset$.} halfspace $H \in \half$ of $\Damu$ such that $\mu(\intr{H})\leq\alpha$, $x\in H$ and $F\subset\bd{H}$. If, in addition, $\mu\in\Atom$, then we can write even $\mu(\intr{H})<\alpha$.
\end{lemma}

Our second lemma is a simple observation about the structure of a simplex, that is a convex hull of $k+1$ points in general position, in the linear space $\R^k$. These $k+1$ points are called the vertices of $S$.

\begin{lemma}\label{separating simplex}
For a simplex $S \subset \R^k$ and any convex set $K \subseteq S$ with non-empty interior there exist $x, y \in K$ and $v \in \Sph[k]$ such that each of the disjoint halfspaces $H_{x,v}$ and $H_{y,-v}$ contains only one vertex of $S$.  
\end{lemma}

\begin{proof}
In this proof, all the vectors are column vectors, and by $A\tr$ we denote the transpose of a matrix $A$. Denote $s_1, \dots, s_{k+1} \in S$ the vertices of $S$. Denote by $a$ any point in the interior of $K$. We first transform both $S$ and $K$ by an affine transform $T \colon \R^k \to \R^k \colon z \mapsto A \, z + b$ for $A\in\R^{k \times k}$ non-singular and $b \in \R^k$ such that $T(s_i) = e_i$ for each $i = 1, \dots, k$ for $e_i$ the $i$-th standard basis vector in $\R^k$, and $T(a) = 0$ is the origin in $\R^k$. Such an affine transform certainly exists, because each full-dimensional simplex in $\R^k$ can be uniquely mapped to any other one using an invertible affine mapping. Because $a \in \intr{K} \subseteq \intr{S}$, the origin $T(a)$ must be contained in the interior of the $T$-image of $S$ defined by $T(S) = \left\{T(z) \colon z \in S\right\}$, meaning that necessarily $T(s_{k+1}) \in (-\infty,0)^k$. Since $K$ is a convex set with $a$ in its interior, also $T(K)$ is convex with $0 = T(a) \in \intr{T(K)}$. Thus, there is a closed ball $B$ centred at the origin with radius $\delta > 0$ small enough so that $B \subseteq T(K) \subseteq T(S)$. For $\widetilde{v} = e_1 \in \Sph[k]$ we have $\left\langle \widetilde{v}, T(s_1) \right\rangle = \left\langle \widetilde{v}, e_1 \right\rangle = 1$, $\left\langle \widetilde{v}, T(s_i) \right\rangle = \left\langle \widetilde{v}, e_i \right\rangle = 0$ for $i=2,\dots,k$, and $\left\langle \widetilde{v}, T(e_{k+1}) \right\rangle < 0$. Take $\widetilde{x} = \delta\, e_1 \in B$ and $\widetilde{y} = - \widetilde{x} \in B$. Then $H_{\widetilde{x}, \widetilde{v}} = \left\{ z \in \R^k \colon \left\langle \widetilde{v}, z \right\rangle \geq \delta \right\}$ contains $e_1 = T(s_1)$ as the only vertex of $T(S)$, and $H_{\widetilde{y}, -\widetilde{v}} = \left\{ z \in \R^k \colon \left\langle \widetilde{v}, z \right\rangle \leq -\delta \right\}$ contains only $T(e_{k+1})$ as the only vertex of $T(S)$. Certainly, also $H_{\widetilde{x}, \widetilde{v}} \cap H_{\widetilde{y}, -\widetilde{v}} = \emptyset$. Now it remains to apply the inverse affine transform $T^{-1} \colon \R^k \to \R^k \colon z \mapsto A^{-1}\left(z-b\right)$ for $A^{-1} \in \R^{k \times k}$ the inverse of $A$, and define $x = T^{-1}(\widetilde{x})$, $y = T^{-1}(\widetilde{y})$, and $v = \left(A\tr e_1\right)/\left\Vert \left(A\tr e_1\right) \right\Vert \in \Sph[k]$. Because $v$ is taken to be the inner normal vector of $T^{-1}(H_{\widetilde{x},\widetilde{v}}) = H_{x,v}$, we indeed found the desired pair of halfspaces $H_{x,v}$ and $H_{y,-v}$.
\end{proof}

We are ready to prove Theorem~\ref{theorem:median dimension}. Recall that $\am = \sup_{x \in \R^d} \depth{x}{\mu}$. Assume for a contradiction that $\dim(\median)= d-1$. Then $\median$ is contained in a hyperplane that determines two different closed halfspaces --- we denote them by $H^+$ and $H^-$, respectively. Take any $w\in\intr{H^+}$ and $q\in\intr{H^-}$. We can consider the set $\median$ itself as a $(d-1)$-dimensional face of $\median$ that satisfies the conditions of Lemma~\ref{lemma:Laketa} for either of the choices $x = w$, or $x = q$. We apply Lemma~\ref{lemma:Laketa} twice, first to $x = w$ and then also to $x = q$. We obtain that $\mu(\intr{H^+}) < \am$ and $\mu(\intr{H^-}) < \am$. Denoting $G^+=\intr{H^+}$, $G^-=\intr{H^-}$ and $A=\bd{H^+} = \bd{H^-}$ we can write 
	\begin{equation}\label{smaller mass}
	\max\left\{\mu(G^+),\mu(G^-)\right\}<\am.
	\end{equation}
Applying Corollary~\ref{cor:boundary depth} to $x\in\median$ and halfspaces $H^+$ and $H^-$ and using~\eqref{smaller mass}, we get that
\begin{equation}\label{depth difference}
\depth{x}{\mu|_{A}}\geq \am-\mu(G)>0\mbox{ for all }x\in\median\mbox{ and }G\in\{G^+,G^-\}.
\end{equation}
Because $\dim(\median)=d-1$ and $\depth{x}{\mu|_A}>0$ for all $x\in\median$, there must exist at least $d$ atoms of $\mu$ in the hyperplane $A$. At the same time, due to our assumption of the atoms of $\mu$ being in general position, there are at most $d$ atoms of $\mu$ in any hyperplane, meaning that $A$ contains exactly $d$ atoms of $\mu$, and these atoms are in general position inside $A$. Consequently,  
	\begin{equation}\label{positive boundary depth}
	\depth{x}{\mu|_A}=1/n\quad\mbox{for all }x\in\median.
	\end{equation}
From~\eqref{depth difference} and~\eqref{positive boundary depth} it follows that $\am>\mu(G)\geq \am-1/n$ for $G\in\{G^+,G^-\}$. Since $\mu$ is an empirical measure with $n$ atoms, the $\mu$-mass of any set can be only a multiple of $1/n$, so it must be that $\mu(G^+)=\mu(G^-)=\am-1/n$. We obtain 
	\begin{equation}\label{mass sum}
	2\left(\am-\frac{1}{n}\right)=\mu(G^+)+\mu(G^-)=1-\mu(A).
	\end{equation}

Since we have shown that there are exactly $d$ atoms of $\mu$ in $A$, it has to be $n\geq d$. From an assumption of our theorem we thus have $n>d$. Then there exists $z\in\R^d\setminus A$ such that $\mu(\{z\})=1/n$. We apply Lemma~\ref{separating simplex} in the subspace $A$ to conclude that there exist $x,y\in\median$ and closed halfspaces $H_{x,v},H_{y,-v}$ in space $A$ such that $\mu|_A(H_{x,v})=\mu|_A(H_{y,-v})=1/n$ and $H_{x,v}\cap H_{y,-v}=\emptyset$. Choose a full-dimensional halfspace $H_{x,u}\in\half$ that meets the conditions $H_{x,u}\cap A=H_{x,v}$ and $z\notin H_{x,u}\cup H_{y,-u}$. Denote $S_x=H_{x,u}\setminus A$ and $S_y=H_{y,-u}\setminus A$. Then $H_{x,u}=H_{x,v}\cup S_x$ and $H_{y,-u}=H_{y,-v}\cup S_y$, so we have 
	\begin{equation}\label{mass1}
	\mu(H_{x,u})=1/n+\mu(S_x),\quad \mu(H_{y,-u})=1/n+\mu(S_y).
	\end{equation}
Note that the sets $S_x$ and $S_y$ are disjoint and $\left(S_x\cup S_y\right)\cap \left(A\cup\{z\}\right)=\emptyset$, so 
	\begin{equation}\label{mass2}
	\mu(S_x)+\mu(S_y)=\mu(S_x\cup S_y)\leq 1-\mu\left(A\cup\{z\}\right)=1-\mu(A)-\frac{1}{n}.
	\end{equation}
Combining \eqref{mass sum}, \eqref{mass1} and \eqref{mass2}, we obtain
	\begin{equation*}	
	\mu(H_{x,u})+\mu(H_{y,-u})=\frac{2}{n}+\mu(S_x)+\mu(S_y)\leq 2\am-\frac{1}{n}.
	\end{equation*}
It follows that $\min\{\mu(H_{x,u}),\mu(H_{y,-u})\}<\am$, a contradiction with our choice $\{x,y\}\subset \median$.
\end{proof}

Theorem~\ref{theorem:median dimension} is valid for empirical measures; an analogous theorem for absolutely continuous measures can be found in \cite[Proposition~3.4]{Small1987}. There, it was shown that for $\mu \in \Meas$ satisfying certain smoothness conditions including the existence of the density, the dimension of the median $\median$ cannot exceed $d - 2$ provided that $d \geq 2$. A version of the latter theorem with weaker conditions, but still requiring smoothness and contiguous support of $\mu \in \Meas$, is given in \cite[Corollary~7]{Laketa_Nagy2021b}. Unlike the proofs for smooth measures, the proof of Theorem~\ref{theorem:median dimension} requires the use of flag halfspaces, which makes the derivation more technical and delicate. Without the assumption of general position, the claim of Theorem~\ref{theorem:median dimension} is not valid. An example of a measure in $\mu \in \Atom[\R^2]$ whose atoms are not in general position but $\dim(\median) = 1$ is given in \cite[Section~2]{Laketa_Nagy2021b}.

Excluding the case of $\dim(\median) = d-1$ for random samples from smooth probability measures, one can ask whether there are other dimensions that the sample median set cannot attain. The answer is negative already in the case of $n = 8$ points sampled randomly from a Gaussian distribution in $\R^3$, as we show in the next example. 

\begin{figure}[htpb]
\includegraphics[width=\twofig\textwidth]{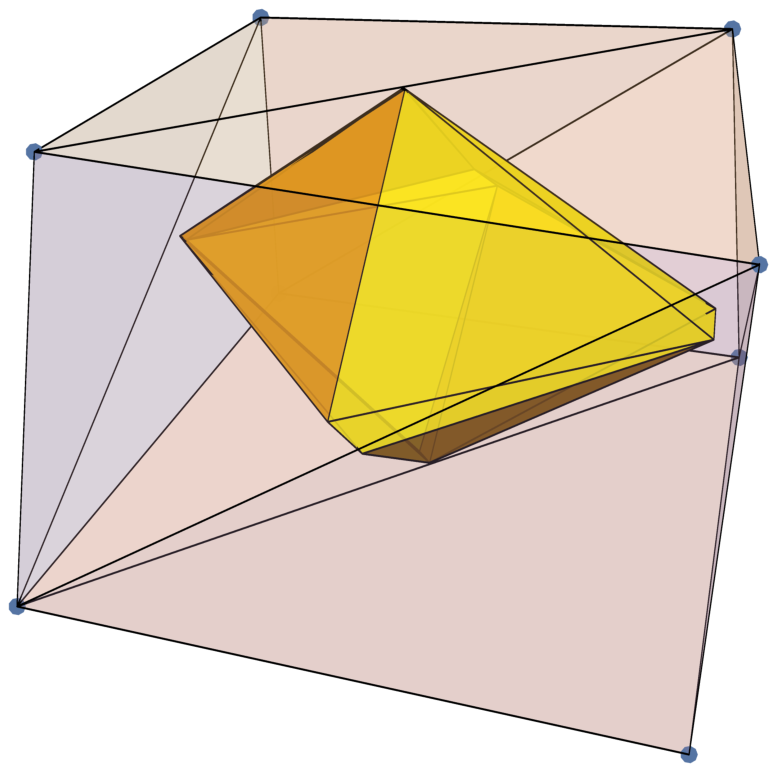} \quad \includegraphics[width=\twofig\textwidth]{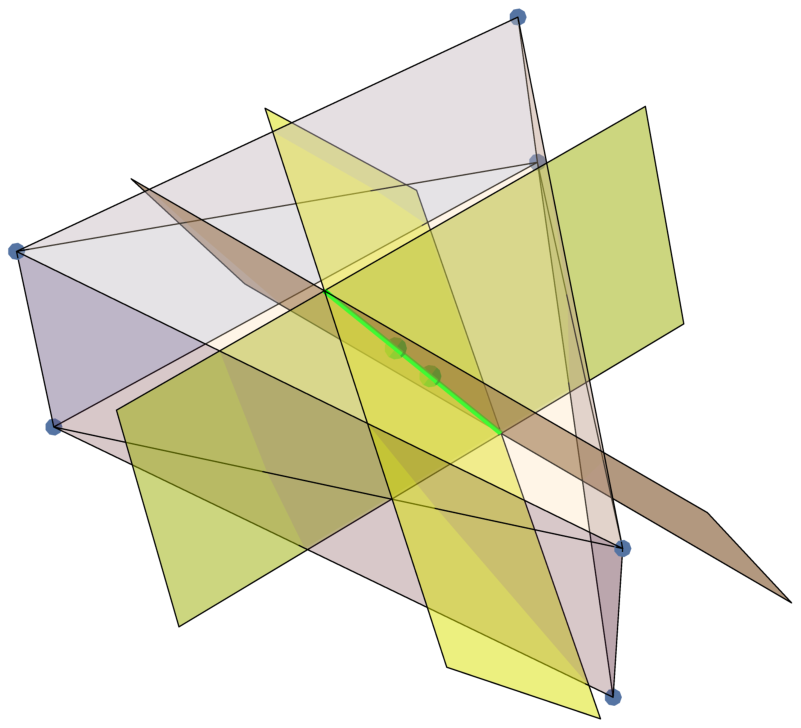}
\includegraphics[width=\twofig\textwidth]{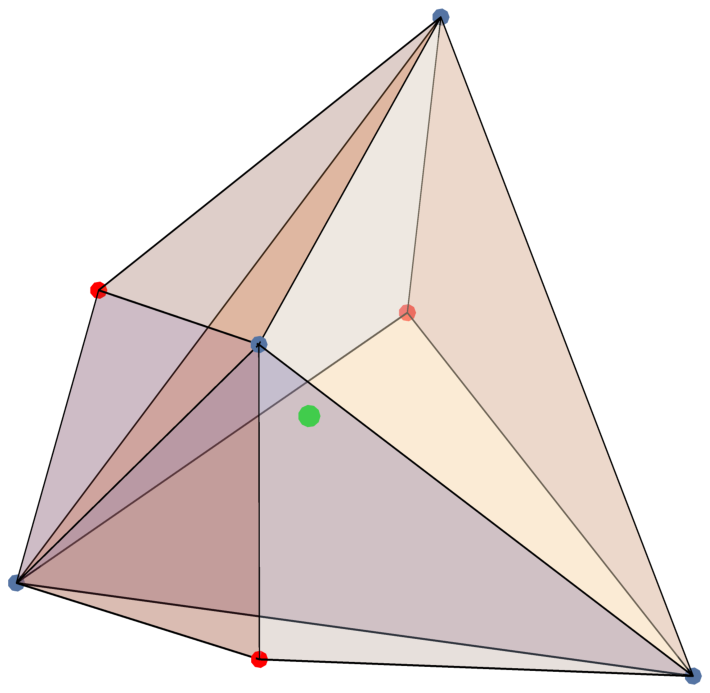} \quad \includegraphics[width=\twofig\textwidth]{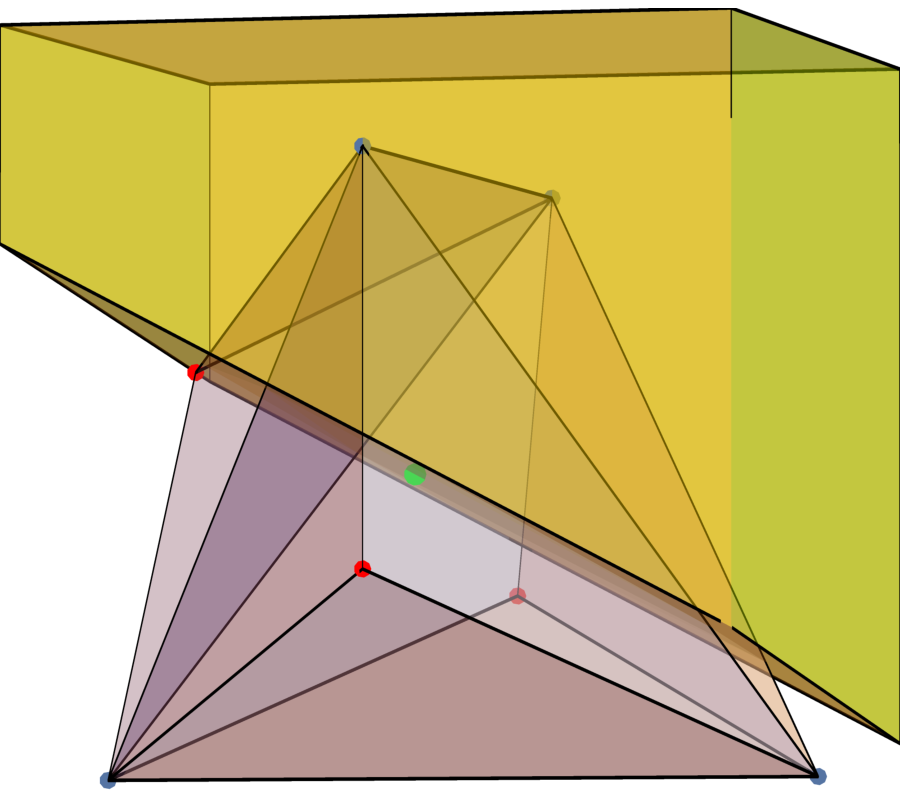}
\caption{Example~\ref{example:medians}. Top left hand panel: Convex hull of the sample points in case $k=3$ (outer polyhedron) and the full-dimensional median set $\median$ (inner coloured polyhedron). Top right hand panel: Convex hull of the sample points in case $k=1$ (outer polyhedron), the median line segment (thick green line segment between the pair of the inner coloured points), and three planes, each separating two sample points from the median set (yellow planes). Bottom panels: Convex hull of the sample points (coloured polyhedron) with the single median (green point). The halfspace in the right hand panel is one minimising halfspace of the median $z$ containing $3$ sample points (the green one and two blue ones). For interactive visualisations see the supplementary \proglang{Mathematica} script.}
\label{figure:medians}
\end{figure}

\begin{example}	\label{example:medians}
For $\nu\in\Meas[\R^3]$ the standard Gaussian probability measure and $X_1, \dots, X_8$ a random sample from $\nu$ with empirical measure $\mu \in \Atom[\R^3]$, the median set $\median$ is of dimension $3$, $1$, or a single-point set, all with positive probability. The claim follows by considering three setups of eight points $x_1, \dots, x_8$ in the space $\R^3$. Denote $k = \dim\left(\median\right)$ and write $\mu \in \Atom[\R^3]$ for the empirical measure of $x_1, \dots, x_8$. The direct computations described below are based on the analysis performed using the \proglang{R} package \pkg{TukeyRegion} \cite{R_TukeyRegion} for evaluation of full-dimensional central regions, and the \textsf{Mathematica} visualisations provided in the script in the online Supplementary Material. Plots of the three setups below are displayed in Figure~\ref{figure:medians}. 
	\begin{itemize}
	\item \textbf{Case $k=3$.} This situation is standard and common. For example, direct computation shows that already for randomly perturbed vertices of a unit cube in $\R^3$, i.e. points in a configuration where the convex hull of $x_1, \dots, x_8$ contains all the eight points on its boundary, possess a full-dimensional polyhedral median set with maximum depth $2/8$.
	\item \textbf{Case $k=1$.} Arrange the points so that $x_1, x_2, x_3$ form vertices of a triangle $T_1$ in a plane, and $x_4, x_5, x_6$ form vertices of a triangle $T_2$ in a plane parallel to that determined by $T_1$, so that the convex hull of $x_1, \dots, x_6$ is a triangular prism in $\R^3$. To obtain points in general position, we perturb the six points slightly. Direct computation shows that for these six points, the halfspace median set is a three-dimensional polyhedron $M$ inside the prism that does not intersect $T_1$ or $T_2$, of points with depth $2/6$. Place the last two points $x_7$ and $x_8$ in the interior of $M$, so that the straight line $l(x_7, x_8)$ between these points intersects both relative interiors of $T_1$ and $T_2$. Note that certainly $\depth{x_7}{\mu} = \depth{x_8}{\mu} = 3/8$, since the two points were placed inside $M$. No point can have depth $4/8$, as in that situation the setup would exhibit halfspace symmetry which is clearly impossible \cite[Proposition~1]{Liu_etal2020}. Finally, projecting all points of $\mu$ into the plane orthogonal to $l(x_7,x_8)$ shows that any point $y \notin l(x_7,x_8)$ can be separated from $l(x_7,x_8)$ by a plane that is parallel to $l(x_7,x_8)$ and contains only two sample points, meaning that $\depth{y}{\mu} \leq 2/8$. The median set of $\mu$ is therefore the line segment between $x_7$ and $x_8$, with depth $3/8$.
	\item \textbf{Case $k=0$.} Consider four points $x_1, \dots, x_4$ forming the vertices of a tetrahedron $T$ (blue points in the bottom panels of Figure~\ref{figure:medians}). Three points $x_5, x_6, x_7 \notin T$ are attached to three different facets of $T$ so that each of these points together with its facet forms another (non-regular) tetrahedron not intersecting $\intr{T}$ (red points in the bottom panels of Figure~\ref{figure:medians}). Finally, a single point $x_8$ is placed strategically inside $T$ into the full-dimensional halfspace median of $x_1, \dots, x_7$. An example is the configuration
		\[	
		\begin{aligned}
		x_1 & = \left(1, 0, -\frac{1}{\sqrt{2}}\right), & x_2 & = \left(-1, 0, -\frac{1}{\sqrt{2}}\right), & x_3 & = \left(0, -1, \frac{1}{\sqrt{2}}\right), & x_4 & = \left(0, 1, \frac{1}{\sqrt{2}}\right), \\
		x_5 & = \left(0, 1, -\frac{1}{4}\right), & x_6 & = \left(\frac{1}{10}, -1, -\frac{1}{4}\right), & x_7 & = \left(\frac{3}{4}, 0, \frac{1}{4}\right), & x_8 & = \left(\frac{1}{10}, \frac{1}{10}, 0\right).
		\end{aligned}
		\]
	These points are in general position in $\R^3$. For the setup of halfspaces $H_{x_8,u_i} \in \half(x_8)$ given by the normal vectors
		\[	u_1 = \left(-\frac{7}{10}, -\frac{3}{10}, -\frac{3}{5}\right), \ u_2 = \left(-\frac{2}{5}, -\frac{1}{10}, \frac{9}{10}\right), \ u_3 = \left(-\frac{1}{5}, \frac{4}{5}, \frac{3}{5}\right), \ u_4 = \left(1, \frac{1}{10}, 0\right)	\]
	we obtain $\mu(H_{x_8,u_i}) = \depth{x_8}{\mu} = 3/8$ for each $i = 1,2,3,4$. At the same time, the union of the open halfspaces $\intr{H_{x_8,u_i}}$ is $\R^3$, meaning that for any $y \ne x_8$ we can find $i$ with $y \in \intr{H_{x_8,u_i}}$, and the shifted closed halfspace $H_{y,u_i} = H_{x_8,u_i} + (y-x_8) \in \half(y)$ necessarily contains at most two atoms of $\mu$. Thus, $\depth{y}{\mu} \leq 2/8$ and the point $x_8$ is the single halfspace median of $\mu$.
	\end{itemize}
The medians in all three cases above are stable in the sense that for a small perturbation of all the sample points, the dimension of the median set remains unchanged. Thus, in each setup and for each $x_i$ we can find a small open ball around $x_i$ such that if $x_i$ is replaced by any element of this ball, the dimension of the new median remains the same. In conclusion, all three cases $k=0,1,3$ occur with positive probability if $x_1, \dots, x_8$ are sampled from any distribution in $\R^3$ with positive density everywhere.\footnote{Note that our example for case $k=1$ happens to disagree with \cite[Theorem~3]{Liu_etal2020}, as for $n = 8$ and $d = 3$ we obtain the maximum depth $\lfloor (n-d+2)/2 \rfloor/n = 3/8$, but the median set is not a single point set. The problem appears to stem from formula (8) in \cite{Liu_etal2020} that is not valid in general.}
\end{example}

\subsection{Computation of the halfspace median in \texorpdfstring{$\R^2$}{R2}}  \label{section:d2}

In dimension $d=2$, Theorem~\ref{theorem:median dimension} leaves only trivial cases: the halfspace median must be either full-dimensional or a singleton, and both situations may occur. But, as we show in our last result below, if $\mu\in\Atom[\R^2]$ has a unique median and $n \ne 4$, then the median must be one of the data points. The case of $n = 4$ data points is trivial and not interesting.\footnote{In the situation $n=4$ and under the assumptions of Lemma~\ref{random samples}, the unique median is almost surely a singleton and is \begin{enumerate*}[label=(\roman*)] \item either the atom contained in the interior of the convex hull of the remaining three sample points; or \item not an atom, but the single point of intersection of the two diagonals of the quadrilateral formed by the convex hull of the atoms.\end{enumerate*}}

\begin{theorem} \label{dimension two} 
Let $\mu\in\Atom[\R^2]$ be an empirical measure with precisely $n$ atoms of mass $1/n$, with $n \ne 4$, that satisfy conditions~\ref{general position} and~\ref{lines intersection} from Lemma~\ref{random samples}. If the halfspace median $\median$ is a single point set, then it must be an atom of $\mu$. In particular, the median set is either full-dimensional, or an atom of $\mu$.
\end{theorem}

\begin{proof}
Suppose without loss of generality that $x=0 \in \R^2$ is the unique median of $\mu$. Assume, for a contradiction, that $\mu(\{0\})=0$. We start with the following observation: For every $v\in \SphO$ there is $w(v)\in\SphO$ that meets the following conditions
	\begin{equation}\label{eq: observation}
	\langle v, w(v) \rangle\geq 0,\ \mu\left(\intr{H_{0,w(v)}}\right)= \am-\frac{1}{n},\mbox{ and }\mu\left(\bd{H_{0,w(v)}}\right)= 2/n.
	\end{equation}
To prove the existence of $w = w(v)$ satisfying \eqref{eq: observation} pick a real sequence $a_i\downarrow 0$ and note that for every $i=1,2,\dots$ we have $a_i v\notin \median$, so there is $w_i\in\SphO$ such that 
	\begin{equation}\label{e1}
	\mu(H_{a_i v,w_i})=\depth{a_i v}{\mu}\leq\am-1/n<\depth{0}{\mu}.
	\end{equation}
The existence of a minimising halfspace $H_{a_i v,w_i} \in \half(a_i v)$ follows from the fact that minimising halfspaces always exist for $\mu\in\Atom$, as we observed in Section~\ref{section:flag halfspace}. Then necessarily $0\not\in H_{a_i v,w_i}$, meaning that $\langle v, w_i \rangle>0$. The sequence $\{w_i\}_{i=1}^\infty\subset \SphO$ is bounded and therefore contains a convergent subsequence $\{w_{i_j}\}_{j=1}^\infty$ with a limit point $w\in\SphO$ that satisfies $\langle v, w\rangle\geq 0$. By the Fatou lemma \cite[Lemma~4.3.3]{Dudley2002} applied to the sets $\intr{H_{0,w}} \subseteq {\lim\inf}_{j\to\infty} \intr{H_{a_{i_j} v,w_{i_j}}}$ and \eqref{e1} we have that 
	\begin{equation}\label{first inequality}
	\mu(\intr{H_{0,w}})\leq\liminf_{j\to\infty} \mu\left(\intr{H_{a_{i_j} v,w_{i_j}}}\right) \leq \am-1/n,
	\end{equation}
 which together with Corollary~\ref{cor:boundary depth} gives us 
	\[	\am=\depth{0}{\mu}\leq \mu(\intr{H_{0,w}})+\depth{0}{\mu|_{\bd{H_{0,w}}}}\leq \am-\frac{1}{n}+\depth{0}{\mu|_{\bd{H_{0,w}}}}.	\]
Therefore, $\depth{0}{\mu|_{\bd{H_{0,w}}}}\geq 1/n$. Because the halfspace median $0$ is not an atom of $\mu$, the condition of general position of the atoms of $\mu$ from part~\ref{general position} of Lemma~\ref{random samples} implies that the straight line $\bd{H_{0,w}}$ contains exactly two atoms of $\mu$ at some points $y,z\in\bd{H_{0,w}}$ such that $0$ is contained in the relatively open line segment formed by $y$ and $z$. Denote by $l_y\subset\bd{H_{0,w}}$ the open halfline centred at $0$ that contains $y$. The flag halfspace $F=\{y\}\cup l_y \cup \intr{H_{0,w}}\in\flag(0)$ then satisfies $\mu(F)=\mu\left(\intr{H_{0,w}}\right)+1/n$. Inequality \eqref{first inequality} implies $\mu(F)\leq \am$, so it must be $\mu(F)=\am$ because of Theorem~\ref{theorem:Pokorny}. Consequently, $\mu\left(\intr{H_{0,w}}\right)=\am-1/n$ and we may take $w(v)=w$. We have proved \eqref{eq: observation}.

Pick any $v\in\SphO$. There exists $u=w(v)\in\SphO$ that satisfies \eqref{eq: observation}. Using the same observation again, we are able to find $u'=w(-u)\in \SphO$ that satisfies \eqref{eq: observation} for $v$ replaced by $-u$. We consider two different cases. 

\textbf{First case:} $u'=-u$. By summing up equalities $\mu\left(\intr{H_{0,u}}\right)= \am-1/n$, $\mu\left(\intr{H_{0,-u}}\right)= \am-1/n$ and $\mu\left(\bd{H_{0,u}}\right)=2/n$ that all follow from \eqref{eq: observation}, we obtain $\am= 1/2$. Consider any infinite line $l$ that passes through the origin and the two open halfplanes $G^+$ and $G^-$ determined by $l$. If $\mu(l)=1/n$, then one of the open halfplanes $G^+$ and $G^-$ is of $\mu$-mass at most $1/2-1/(2n)$. Assume that $\mu(G^+)\leq 1/2-1/(2n)$. Because $l$ contains only one atom of $\mu$ that is not at the origin, there is a flag halfspace $F\in \flag(0)$ composed of $G^+$ and the relatively closed halfline in $l$ starting at $0$ that does not contain atoms. Then $\mu(F)=\mu(G^+)\leq 1/2-1/(2n)<\am$, a contradiction with $\mu(l) = 1/n$. Due to the assumption of general position of atoms from part~\ref{general position} of Lemma~\ref{random samples}, we know that $\mu(l)\leq 2/n$, so $\mu(l)$ can take only one of the two possible values: either $0$ or $2/n$. Because of our assumption from part~\ref{lines intersection} of Lemma~\ref{random samples}, there however cannot be three different lines determined by pairs of sample points that all intersect in the origin. This means that for only at most two lines $l$ in $\R^2$ passing through the origin, the $\mu$-mass of $l$ can be $2/n$; all the other lines that we now consider have null $\mu$-mass (given that we have already excluded the case $\mu(l) = 1/n$). This leaves only two possibilities: either $n=2$, or $n=4$. If $n=2$, then the median set $\median$ is the line segment determined by the only two atoms of $\mu$, and therefore it is one-dimensional. Only the case $n=4$, not covered by the statement of this theorem, remains.

\textbf{Second case:} $u'\neq-u$. There exists a closed halfspace $H_{0,v'}$ whose boundary passes through the origin that does not contain any of the points $u$ and $u'$. Let $\tilde{u}=w(v')$ be the unit vector that satisfies \eqref{eq: observation} with $v=v'$. Directly by \eqref{eq: observation}, each of the three different lines $\bd{H_{0,u}}$, $\bd{H_{0,u'}}$ and $\bd{H_{0,\tilde{u}}}$ contains two atoms of $\mu$, a contradiction with our assumption from part~\ref{lines intersection} of Lemma~\ref{random samples}.

The last part of the statement of Theorem~\ref{dimension two} follows directly from Theorem~\ref{theorem:median dimension}.
\end{proof}

Theorems~\ref{theorem:median dimension} and~\ref{dimension two} fully justify the algorithmic procedure from \cite{Liu_etal2019} and \cite{Fojtik_etal2022} for finding the halfspace medians of samples from smooth probability distributions in $\R^2$. If the median set is full-dimensional, the algorithm from \cite{Liu_etal2019} implemented in the \proglang{R} package \pkg{TukeyRegion} \cite{R_TukeyRegion} finds the median set exactly, as proved in \cite{Fojtik_etal2022}. If the median is not full-dimensional, we conclude that it has to be a single sample point, and evaluation of the maximum halfspace depth of all sample points gives the unique halfspace median. 

In dimension $d > 2$, the situation with possible less-than-full-dimensional halfspace medians appears to be much more convoluted, as demonstrated already in Example~\ref{example:medians}. Our proof technique from Theorem~\ref{dimension two} does not extend directly to $d>2$. One might, however, conjecture that in accordance with Theorem~\ref{dimension two}, a less-than-full-dimensional median of a dataset in general position must contain at least one atom of $\mu\in\Atom$. Our final example shows that this is not true: a configuration of points in general position without an atom in the halfspace median set is indeed possible. 

\begin{example}\label{line segment}
Consider a dataset of $n = 8$ points in $\R^3$ given by
	\[	
	\begin{aligned}
	x_1 & = \left(-1, \frac{1}{3}, -\frac{2}{3}\right), & x_2 & = \bigg(1, 0, -1\bigg), & x_3 & = \left(0, \frac{3}{2}, -1\right), & x_4 & = \left(-\frac{1}{2}, 0, 1\right), \\
	x_5 & = \left(1, 0, \frac{4}{3}\right), & x_6 & = \bigg(0, 2, 1\bigg), & x_7 & = \left(-\frac{1}{3}, \frac{1}{2}, -2\right), & x_8 & = \left(\frac{1}{3}, \frac{1}{2}, 2\right).
	\end{aligned}
	\]
Similarly as in case $k=1$ in Example~\ref{example:medians}, points $x_1, \dots, x_6$ are perturbed vertices of a triangular prism. Points $x_7$ and $x_8$ determine a line segment that passes through both triangular bases of that prism. The dataset is in general position. A direct computation performed in \proglang{Mathematica}, provided in the script in the online Supplementary Material, confirms that the sample halfspace median set of this dataset attains depth $3/8$, and the median set consists of the line segment between points $\left(0,1/2,0\right)$ and $\left(3/44, 1/2, 9/22\right)$. This median line segment lies strictly in the relative interior of the straight line between points $x_7$ and $x_8$, and does not contain any atoms of the corresponding measure $\mu \in \Atom[\R^3]$. Thus, it is possible that in dimension $d>2$, a less-than-full-dimensional median set contains no data points. For a visualisation of our dataset and its median set see Figure~\ref{figure:line segment}.
\end{example}

\begin{figure}[htpb]
\includegraphics[width=\twofig\textwidth]{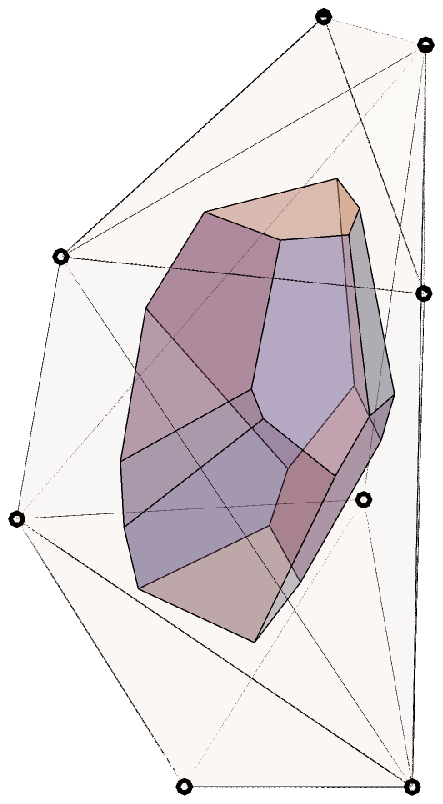} \quad \includegraphics[width=\twofig\textwidth]{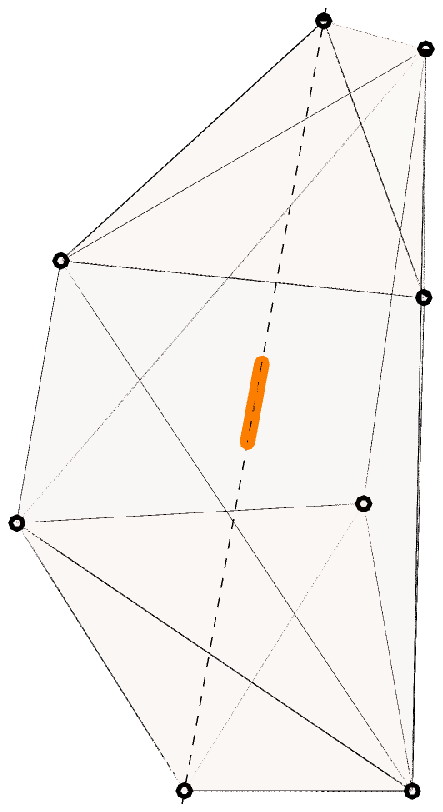}
\caption{Example~\ref{line segment}. Left hand panel: Convex hull of the sample points (outer polyhedron) and the full-dimensional central region $\Damu$ with $\alpha = 2/8$ (inner coloured polyhedron). Right hand panel: Convex hull of the sample points (outer polyhedron), the median line segment of depth $\am = 3/8$ (thick orange line segment), and the line between data points $x_7$ and $x_8$ (dashed black line). The halfspace median line segment forms a piece of the dashed black line, and is contained inside the region $D_{2/8}(\mu)$ from the left hand panel. For interactive visualisations see the supplementary \proglang{Mathematica} script.}
\label{figure:line segment}
\end{figure}

In Example~\ref{line segment}, we constructed a dataset in general position in dimension $d=3$, with a one-dimensional median set. We do not know an example of a dataset sampled from a distribution with a density in $\R^d$, $d>2$, with a unique (zero-dimensional) halfspace median that is not a data point. The higher dimensional situation therefore deserves further investigation. 

\subsection*{Acknowledgement}
P.~Laketa was supported by the OP RDE project ``International mobility of research, technical and administrative staff at the Charles University", grant CZ.02.2.69/0.0/0.0/18\_053/0016976. The work of S.~Nagy was supported by Czech Science Foundation (EXPRO project n. 19-28231X).


\begin{thebibliography}{}

\bibitem[Barber and Mozharovskyi, 2022]{R_TukeyRegion}
Barber, C. and Mozharovskyi, P. (2022).
\newblock {\em TukeyRegion: {T}ukey region and median}.
\newblock R package version 0.1.5.5.

\bibitem[Chernozhukov et~al., 2017]{Chernozhukov_etal2017}
Chernozhukov, V., Galichon, A., Hallin, M., and Henry, M. (2017).
\newblock Monge-{K}antorovich depth, quantiles, ranks and signs.
\newblock {\em Ann. Statist.}, 45(1):223--256.

\bibitem[Donoho and Gasko, 1992]{Donoho_Gasko1992}
Donoho, D.~L. and Gasko, M. (1992).
\newblock Breakdown properties of location estimates based on halfspace depth
  and projected outlyingness.
\newblock {\em Ann. Statist.}, 20(4):1803--1827.

\bibitem[Dudley, 2002]{Dudley2002}
Dudley, R.~M. (2002).
\newblock {\em Real analysis and probability}, volume~74 of {\em Cambridge
  Studies in Advanced Mathematics}.
\newblock Cambridge University Press, Cambridge.

\bibitem[Dyckerhoff and Mozharovskyi, 2016]{Dyckerhoff_Mozharovskyi2016}
Dyckerhoff, R. and Mozharovskyi, P. (2016).
\newblock Exact computation of the halfspace depth.
\newblock {\em Comput. Statist. Data Anal.}, 98:19--30.

\bibitem[Fojt\'ik et~al., 2022]{Fojtik_etal2022}
Fojt\'ik, V., Laketa, P., Mozharovskyi, P., and Nagy, S. (2022).
\newblock On exact computation of {T}ukey depth central regions.
\newblock {\em arXiv preprint arXiv:2208.04587}.

\bibitem[Laketa and Nagy, 2022a]{Laketa_Nagy2021b}
Laketa, P. and Nagy, S. (2022a).
\newblock Halfspace depth for general measures: the ray basis theorem and its
  consequences.
\newblock {\em Statist. Papers}, 63(3):849--883.

\bibitem[Laketa and Nagy, 2022b]{Laketa_Nagy2022}
Laketa, P. and Nagy, S. (2022b).
\newblock Partial reconstruction of measures from halfspace depth.
\newblock In {\em Proceedings of {CLADAG2021}}, Stud. Classification Data Anal.
  Knowledge Organ. Springer, Cham.
\newblock To appear.

\bibitem[Laketa et~al., 2022]{Laketa_etal2022simple}
Laketa, P., Pokorn\'y, D., and Nagy, S. (2022).
\newblock Simple halfspace depth.
\newblock Under review.

\bibitem[Liu et~al., 2020]{Liu_etal2020}
Liu, X., Luo, S., and Zuo, Y. (2020).
\newblock Some results on the computing of {T}ukey's halfspace median.
\newblock {\em Statist. Papers}, 61(1):303--316.

\bibitem[Liu et~al., 2019]{Liu_etal2019}
Liu, X., Mosler, K., and Mozharovskyi, P. (2019).
\newblock Fast computation of {T}ukey trimmed regions and median in dimension
  {$p>2$}.
\newblock {\em J. Comput. Graph. Statist.}, 28(3):682--697.

\bibitem[Mass{\'e}, 2004]{Masse2004}
Mass{\'e}, J.-C. (2004).
\newblock Asymptotics for the {T}ukey depth process, with an application to a
  multivariate trimmed mean.
\newblock {\em Bernoulli}, 10(3):397--419.

\bibitem[Mizera and Volauf, 2002]{Mizera_Volauf2002}
Mizera, I. and Volauf, M. (2002).
\newblock Continuity of halfspace depth contours and maximum depth estimators:
  diagnostics of depth-related methods.
\newblock {\em J. Multivariate Anal.}, 83(2):365--388.

\bibitem[Mosler and Mozharovskyi, 2022]{Mosler_Mozharovskyi2022}
Mosler, K. and Mozharovskyi, P. (2022).
\newblock Choosing among notions of multivariate depth statistics.
\newblock {\em Statist. Sci.}, 37(3):348--368.

\bibitem[Nagy et~al., 2019]{Nagy_etal2019}
Nagy, S., Sch\"{u}tt, C., and Werner, E.~M. (2019).
\newblock Halfspace depth and floating body.
\newblock {\em Stat. Surv.}, 13:52--118.

\bibitem[Rousseeuw and Ruts, 1999]{Rousseeuw_Ruts1999}
Rousseeuw, P.~J. and Ruts, I. (1999).
\newblock The depth function of a population distribution.
\newblock {\em Metrika}, 49(3):213--244.

\bibitem[Schneider, 2014]{Schneider2014}
Schneider, R. (2014).
\newblock {\em Convex bodies: the {B}runn-{M}inkowski theory}, volume 151 of
  {\em Encyclopedia of Mathematics and its Applications}.
\newblock Cambridge University Press, Cambridge, expanded edition.

\bibitem[Small, 1987]{Small1987}
Small, C.~G. (1987).
\newblock Measures of centrality for multivariate and directional
  distributions.
\newblock {\em Canad. J. Statist.}, 15(1):31--39.

\bibitem[Struyf and Rousseeuw, 1999]{Struyf_Rousseeuw1999}
Struyf, A. and Rousseeuw, P.~J. (1999).
\newblock Halfspace depth and regression depth characterize the empirical
  distribution.
\newblock {\em J. Multivariate Anal.}, 69(1):135--153.

\bibitem[Tukey, 1975]{Tukey1975}
Tukey, J.~W. (1975).
\newblock Mathematics and the picturing of data.
\newblock In {\em Proceedings of the {I}nternational {C}ongress of
  {M}athematicians ({V}ancouver, {B}. {C}., 1974), {V}ol. 2}, pages 523--531.
  Canad. Math. Congress, Montreal, Que.

\bibitem[Zuo and Serfling, 2000]{Zuo_Serfling2000}
Zuo, Y. and Serfling, R. (2000).
\newblock General notions of statistical depth function.
\newblock {\em Ann. Statist.}, 28(2):461--482.

\end{thebibliography}
\end{document}